\documentclass[a4paper]{article}
\usepackage{amsmath}

\newtheorem{theorem}{Theorem}

\newtheorem{proposition}{Proposition}

\newtheorem{example}{Example}
\newenvironment{proof}[1][Proof]{\noindent\textbf{#1.} }{\ \rule{0.5em}{0.5em}}

\begin{document}

\title{\textbf{Cyclic Codes over Some Finite Rings}}
\author{Mehmet \"{O}zen, Murat G\"{u}zeltepe \\
{\small Department of Mathematics, Sakarya University, TR54187 Sakarya,
Turkey}}
\date{}
\maketitle

\begin{abstract}
In this paper cyclic codes are established with respect to the
Mannheim metric over some finite rings by using Gaussian integers
and the decoding algorithm for these codes is given.
\end{abstract}




\section{Introduction }

Mannheim metric, which was initially put forward by Huber 1994,
has been used in many papers so far \cite{1,2,3,4,5}. In 1994,
Huber defined the Mannheim metric and the Mannheim weight over
Gaussian integers and, eventually, he obtained the linear codes
which can correct errors of Mannheim weight one in \cite{1}.
Moreover, some of these codes which are suited for quadrature
amplitude modulation (QAM)-type modulations were considered by
Huber. Later, Huber transferred these codes, which he obtained by
using the Mannheim metric and Gaussian integers, into
Eisenstein-Jacobi integers in \cite{2}. In 1997, Huber proved the
MacWilliams theorem for two-dimensional modulo metric (Mannheim
metric) \cite{3}. In 2001, Neto obtained new codes over Euclidean
domain  $ \mathcal {Q}\sqrt d  $ , where $ d =  - 1, - 2, - 3, -
7, - 11 $  , in \cite{4} by using the Mannheim metric given in
\cite{1,2}. In 2004, Fan and Gao obtained one-error correcting
linear codes by using a novel Mannheim weight over finite
algebraic integer rings of cyclotomic fields \cite{5}. In our
study, the codes in \cite{1} are transferred into some finite ring
by using the Mannheim metric and Gaussian integers. Since these
codes are transferred from finite fields into finite rings, there
occurs some difference in decoding process and we also mention
these differences.\\

Section II is organized as follows. In Proposition 1, the
necessary algebraic background is revealed in order to obtain
cyclic codes. In Theorem 2, it is shown how to obtain cyclic codes
by utilizing Proposition 1, 2 and 3. In Proposition 4, the
algebraic background, which is essential for obtaining cyclic code
over the other finite rings, is arranged and in Theorem 3, it is
shown how to obtain cyclic codes over the other finite rings.

\section{Cyclic Codes over Gaussian Integer}
A Gaussian integer is a complex number whose real and imaginary
parts are both integers. Let $ G $ denote the set of all Gaussian
integers and let $ G_\pi $ denote residue class of $ G $
modulo$\pi$, where $ \pi  = a + ib $ is a Gaussian prime integer
and $ p $ is a prime integer such that $ p = a^2  + b^2  = 4n + 1$
The modulo function $ GF(p) \to G_\pi $ is defined by
\begin{equation}  \label{eq:1}  \mu (g) =
g - [\frac{{g.\pi ^* }}{p}].\pi,
\end{equation}
where $ GF(p) $ is a finite field with $p$ elements. elements. In
(\ref{eq:1}), the symbol of [.] is rounding to the closest
integer. The rounding of Gaussian integer can be done by rounding
the real and imaginary parts separately to the closest integer. In
view of equation (\ref{eq:1}), $ G_\pi $ is isomorphic to $ Z_p $,
where $ Z_p $ is residue class of the set $Z$ of all integers
modulo $p$. Let $ \alpha$ and $\beta $ be elements in $ G_\pi $
then the Mannheim weight of  $ \gamma $ is defined by $ w_M
(\gamma ) = \left| {{\mathop{\rm Re}\nolimits} (\gamma )} \right|
+ \left| {{\mathop{\rm Im}\nolimits} (\gamma )} \right| $,  where
$ \gamma  = \alpha  - \beta \;\bmod \pi $. Since the linear codes
are linear code, the Mannheim distance between $ \alpha $ and
$\beta$ is $ d_M (\alpha ,\beta ) = w_M (\gamma )$ \cite{1}.

\begin{theorem}
\label{thm:one} If $a$ and $b$ are relatively prime integers, then
$ {{G = Z[i]} \mathord{\left/ {\vphantom {{G = Z[i]} {\left\langle
{a + ib} \right\rangle  \cong Z_{a^2  + b^2 } }}} \right.
 \kern-\nulldelimiterspace} {\left\langle {a + ib} \right\rangle  \cong Z_{a^2  + b^2 } }}
$ \cite{6}.
\end{theorem}

\begin{proposition} Let $\pi  = a + bi $ be a prime in $G$ and let
$ p > 2 $ be prime element in $Z$ such that $ p = a^2  + b^2  = 4n
+ 1$. If $G$ is a generator of $ G_{\pi ^2 }^ * $, then $ g^{\phi
(p^2 )/4}  \equiv i\;\bmod \pi ^2 $ ( or $ g^{\phi (p^2 )/4}
\equiv  - i\;\bmod \pi ^2 )$.

\end{proposition}

\begin{proof}
If $ \left| \pi  \right| = 4n + 1 $ is a prime integer in $Z$ ,
the real and imaginary parts of $ \pi ^2 $ are relatively primes,
where the symbol $ \left|  \cdot  \right| $  denotes modulo of a
complex number. So, $ G_{\pi ^2 } $ is isomorphic to $ Z_{p^2 } $
(See Theorem 1). If $g$ is a generator of $ G_{\pi ^2 }^ * $ ,
then $ g,\,g^2 ,\,...,\,g^{\phi (p^2 )} \;\bmod \pi ^2 $
constitute a reduced residue system. Therefore there is a positive
integer $k$ as $ g^k  \equiv i\;\bmod \pi ^2 $ ($ g^k  \equiv  -
i\;\bmod \pi ^2 $, where $ 1 \le k \le \phi (p^2 ) $. Hence, we
can infer $ g^{4k}  \equiv 1\;\bmod \pi ^2 $. Since $ \left. {\phi
(p^2 )} \right|4k $ and $ 4 \le 4k \le 4\phi (p^2 )$ , we obtain
$\phi (p^2 ) = k,\phi (p^2 ) = 2k{\rm}$ or $\phi (p^2 ) = 4k{\rm }
$. If $\phi (p^2 ) = k$ was equal to $k$ or $2k$ , we should have
$ \left. {\pi ^2 } \right|i - 1 $ or $ \left. {\pi ^2 }
\right|2\;\;(\left. {\pi ^2 } \right| - i - 1) $, but this would
contradict the fact that  $ \left| {\pi ^2 } \right|^2  > 2 $.
\end{proof}

\begin{proposition} Let $\pi  = a + ib $ be a prime in $G$ and let
$ p > 2 $ be prime in $Z$ such that $ p = a^2  + b^2  = 4n + 1 $.
If $g$  is a generator of $ G_{\pi ^k }^ * $, the $ g^{\phi (p^k
)/4}  \equiv i\;\bmod \pi ^k $ or ($ g^{\phi (p^k )/4}  \equiv -
i\;\bmod \pi ^k $).
\end{proposition}
\begin{proof}This is immediate from Proposition 1.
\end{proof}

\begin{proposition}Let $\pi  = a + ib $ be a prime in $G$ and let
$ p > 2 $ be prime in $Z$ such that $ p = a^2  + b^2  = 4n + 1 $.
If $g$  is a generator of $ G_{\pi ^k }^ * $ and $ g^{\phi (p^2
)/4}  \equiv i\;\bmod \pi ^2 $, then $-g$ also becomes a generator
of $ G_{\pi ^2 }^ * $ such that $ ( - g)^{\phi (p^2 )/4} \equiv -
i\;\bmod \pi ^2 $.
\end{proposition}
\begin{proof}$ g^{\phi (p^2 )/4}  \equiv i\;\bmod \pi ^2 $ implies
that$ ( - g)^{\phi (p^2 )/4}  \equiv  - i\;\bmod \pi ^2 $ since $
\phi (p^2 ) = 4n(4n + 1) $ and $n$  is an odd integer.
\end{proof}

\begin{theorem} Let $p > 2 $ is a prime in $Z$  and $ \pi  = a + ib
$ is a prime in  $G$ such that $ p = a^2  + b^2  = 4n + 1 $ ($
a,b,n \in Z $) , then cyclic codes of length $ \phi (p^2 )/4 $ and
$ \phi (p^2 )/2 $ are generated over the ring  $ G_{\pi ^2 } $
whose the generator polynomial are of the first and second degree,
respectively.
\end{theorem}
\begin{proof}There is an element $g$  of $G_{\pi ^2 } $ and $G_{\pi ^2 }^ *
$ is generated by $g$  since $ Z_{p^2 } $ is isomorphic to $G_{\pi
^2 } $. We know that $ g^{\phi (p^2 )/4}  \equiv i\;\bmod \pi ^2 $
implies that $ ( - g)^{\phi (p^2 )/4}  \equiv  - i\;\bmod \pi ^2 $
from Proposition 3. Hence  $ x^{{{\phi (p^2 )} \mathord{\left/
 {\vphantom {{\phi (p^2 )} 4}} \right.
 \kern-\nulldelimiterspace} 4}}  - i
$ and $ x^{{{\phi (p^2 )} \mathord{\left/
 {\vphantom {{\phi (p^2 )} 4}} \right.
 \kern-\nulldelimiterspace} 4}}  + i
$ are factored as $ (x - g)Q(x) $ $ \bmod \pi ^2 $ for $ x = g $
and $ (x + g)R(x)\;\bmod \pi ^2 $ for $ x = -g $, respectively,
where $ Q(x) $ and $ R(x) $ are the polynomials in the
indeterminate $X$  with coefficients in $ G_{\pi ^2 } $. Moreover,
$ x^{{{\phi (p^2 )} \mathord{\left/
 {\vphantom {{\phi (p^2 )} 2}} \right.
 \kern-\nulldelimiterspace} 2}}  + 1
$ can be factored as $ (x - g)(x + g)A(x) $ $ \bmod \pi ^2 $,
where $ A(x) $ is the polynomials in the indeterminate $X$ with
coefficients in $ G_{\pi ^2 } $.Furthermore all components of any
row of generator matrix do not consist of zero divisors since the
generator polynomial would be selected as a monic polynomial.
\end{proof}\\

We now explain how to construct cyclic codes over the other finite
rings.\\

Denote $ Z_n^ * $ by the set of multiplicative inverse elements of
$ Z_n $. If $ k \ge 1 $ and $ \left. k \right|n $ then the set $
Z_n^ *  (k) $ is a subgroup of $ Z_n^ * $, where $ Z_n^* (k) =
\left\{ {x \in Z_n^* :x \equiv 1\begin{array}{*{20}c}
   {\bmod k}  \\
\end{array}} \right\}$. If  $s$ and $t$  are relatively prime numbers,
then $ Z_{st}^ *  (s) \cong Z_t^ *  ,Z_{st}^ *  (t) \cong Z_s^ *
$.
\begin{proposition}Let $p_1 $ and $p_2 $ be odd primes and let $\pi _1  = a + bi
$ and $ \pi _2  = c + di $ be prime Gaussian integers, where $ p_1
\ne p_2 $ and $ p_1  = a^2  + b^2  = 4n_1  + 1 $ and $ p_2  = c^2
+ d^2  = 4n_2  + 1 $ ($ a,b,c,d,n_1 ,n_2  \in Z $). If $ \pi _1 $
and $ \pi _2 $ are Gaussian integers, there exist any elements
$e$, $f$ of $ G_{\pi _1 \pi _2 }^ * $ satisfying $ e^{\phi (p_2 )}
\equiv 1\;\bmod \pi _1 \pi _2 $ and $ f^{\phi (p_1 )}  \equiv
1\;\bmod \pi _1 \pi _2 $.

\end{proposition}
\begin{proof}Let $p_1 $ and $p_2 $ be distinct odd primes. Then $p_1 $ and $p_2
$ are relatively primes. Since $s$  and $t$  are relatively prime
numbers, $p_1 $ and $p_2 $ can be selected as $s$  and $t$,
respectively, that is,  $ s = a^2  + b^2  = 4n_1  + 1 $ and $ t =
c^2  + d^2  = 4n_2  + 1 $. Using (\ref{eq:1}), we have $ Z_s \cong
G_{\pi _1 } $ and $ Z_t  \cong G_{\pi _2 } $. It is clear that
$Z_{s.t}  \cong G_{\pi _1 .\pi _2 } $ from Theorem 1. Thus, we
have $ G_{\pi _1 .\pi _2 }^ *  (\pi _1 ) \cong Z_{s.t}^ *  (s)
\cong Z_t^ *   \cong G_{\pi _2 }^ * $. $ G_{\pi _2 }^ * $ is a
cyclic group because $ \pi _2 $ is a prime Gaussian integer. So, $
G_{\pi _1 .\pi _2 }^ *  (\pi _1 ) $ has a generator. Let's call
this generator $e$ . Then $ e^{\phi (p_2 )}  \equiv
1\;\begin{array}{*{20}c}
   {\bmod \pi _1 }  \\
\end{array}.\pi _2$. In the similar way, $G_{\pi _1 .\pi _2 }^ *  (\pi _2 )
$ has a generator, let's call $f$ . Then $ f^{\phi (p_1 )} \equiv
1\;\begin{array}{*{20}c}
   {\bmod \pi _1 }  \\
\end{array}.\pi _2$ since $G_{\pi _1 .\pi _2 }^ *  (\pi _2 ) $ is  isomorphic
to $ G_{\pi _1 }^ * $.
\end{proof}
\begin{proposition}Let $p_k $ be $k$ distinct prime integer, where
$p_k $  is prime odd integer, $ \pi _k  = a_k  + ib_k $ and $ p_k
= a_k^2  + b_k^2 = 4n_k  + 1 $ for $ k = 1,2,...,m $.  . So, there
exists an element $ e_k $ of $ G_{\pi _1 .\pi _2 ...\pi _m }^ * $
such that $ e_k^{\phi (p_k )}  \equiv 1\;\bmod \pi _1 .\pi _2
...\pi _m $.
\end{proposition}

\begin{proof} This is immediate from Proposition 4.
\end{proof}

\begin{theorem}Let $p_1 $ and $p_2 $ be distinct prime odd integers in $Z$  and
let $ \pi _1  = a + bi $ and $ \pi _2  = c + di $ be prime
Gaussian integers in $G$, where $ p_1  = a^2  + b^2  = 4n_1  + 1
$, $ p_2  = c^2  + d^2  = 4n_2  + 1 $, $ n_1 ,n_2  \in Z $. Then
there exists a cyclic code of  length $ \phi (p_1 ) $ and $ \phi
(p_2 ) $ over the ring $ G_{\pi _1 .\pi _2 } $. The generator
polynomial of this cyclic code is a first degree monic polynomial.
\end{theorem}

\begin{proof}From Proposition 4, $x^{\phi (p_2 )}  - 1 $ can be factored as
$ (x - e).D(x)\bmod \pi _1 \pi _2 $ since $ e^{\phi (p_2 )} \equiv
1\;\begin{array}{*{20}c}
   {\bmod \pi _1 }  \\
\end{array}.\pi _2$. If we take the generator polynomial as $
g(x) = x - e $, then the generator polynomial $g(x)$ generates the
generator matrix, whose any row of all components do not consist
of zero divisors.
\end{proof}\\

To illustrate the construction of cyclic codes over some finite
rings, we consider examples as follows.

\begin{example}The polynomial $ x^{10}  + 1 $ factors over the ring
$ G_{3 + 4i} $ as $ (x - 2).(x - 1 + i)A(x)$, where $A(x)$ is a
polynomial in $ G_{3 + 4i} \left[ X \right] $. If  the generator
polynomial $g(x)$  is taken as $ x^2  + (1 - 2i)x + ( - 2 + i) $,
then the generator matrix and the parity check matrix are as
follows, respectively.\\$$ G = \left( {\begin{array}{*{20}c}
   { - 2 + i} & {1 - 2i} & 1 & 0 & 0 & 0 & 0 & 0 & 0 & 0  \\
   0 & { - 2 + i} & {1 - 2i} & 1 & 0 & 0 & 0 & 0 & 0 & 0  \\
   0 & 0 & { - 2 + i} & {1 - 2i} & 1 & 0 & 0 & 0 & 0 & 0  \\
   0 & 0 & 0 & { - 2 + i} & {1 - 2i} & 1 & 0 & 0 & 0 & 0  \\
   0 & 0 & 0 & 0 & { - 2 + i} & {1 - 2i} & 1 & 0 & 0 & 0  \\
   0 & 0 & 0 & 0 & 0 & { - 2 + i} & {1 - 2i} & 1 & 0 & 0  \\
   0 & 0 & 0 & 0 & 0 & 0 & { - 2 + i} & {1 - 2i} & 1 & 0  \\
   0 & 0 & 0 & 0 & 0 & 0 & 0 & { - 2 + i} & {1 - 2i} & 1  \\
\end{array}} \right)
,$$\\ $$ H = \left( {\begin{array}{*{20}c}
   1 & { - (1 - 2i)} & { - ( - 2 + i)} & { - (2 + i)} & { - (1 + i)} & { - (2 + i)} & {3i} & { - 1 + i} & {2i} & 0  \\
   0 & 1 & { - (1 - 2i)} & { - ( - 2 + i)} & { - (2 + i)} & { - (1 + i)} & { - (2 + i)} & {3i} & { - 1 + i} & {2i}  \\
\end{array}} \right)\\
.$$ Let us assume that at the receiving end we get the vector $ r
= \left( {\begin{array}{*{20}c}
   { - 2 + i} & {1 - 2i} & 1 & i & 0 & 0 & 0 & 0 & 0 & 0  \\
\end{array}} \right)$. First we compute the syndrome $S$  as follows:
$$S = \frac{{ix^3  + x^2  + (1 - 2i)x + ( - 2 + i)}}{{x^2  + (1 -
2i)x + ( - 2 + i)}} = (1 + 2i)x + (2 + i) .$$ Therefore, from
Table I, it is seen that the syndrome $ S \equiv ix^3 $. Notice
that first we compute the syndrome of the received vector to be
decoded. If this syndrome disappears in Table I, then its
associates check. Thus, the received vector $r$  is decoded as $
c(x) = r(x) - ix^3  = x^2  + (1 - 2i)x + ( - 2 + i) $. Finally we
get\\ $ c = \left( {\begin{array}{*{20}c}
   { - 2 + i} & {1 - 2i} & 1 & 0 & 0 & 0 & 0 & 0 & 0 & 0  \\
\end{array}} \right)
$.

\end{example}

\begin{example} Let $p_1  = 5 $, $p_2  = 13 $. Let the generator polynomial
$g(x)=x-3-i$ and the parity check polynomial $ h(x) = [x^3  + (3 +
i)x^2  + (4 - i)x + 2 - 2i] $ be. Then we obtain the generator
matrix $G$  and parity check matrix $H$  as follows, respectively.
$$
G = \left( {\begin{array}{*{20}c}
   { - 3 - i} & 1 & 0 & 0  \\
   0 & { - 3 - i} & 1 & 0  \\
   0 & 0 & { - 3 - i} & 1  \\
\end{array}} \right),
$$\\ $$
H = \left( {\begin{array}{*{20}c}
   1 & {3 + i} & {4 - i} & {2 - 2i}  \\
\end{array}} \right)
.$$\\ Assume that received vector is $ r = \left(
{\begin{array}{*{20}c}
   { - 3 - i} & 1 & i & 0  \\
\end{array}} \right)$. We compute the syndrome as $$
\frac{{r(x)}}{{g(x)}} = \frac{{ix^2  + x - 3 - i}}{{x - 3 - i}} =
(x - 3 - i)(ix + 3i) + (1 + 4i) .$$ Since $ 1 + 4i \equiv x^2 .i
$,  the vector $r(x)$  is decoded as $ c(x) = r(x) - ix^2  = x - 3
- i $. In Table II, coset leaders and its syndromes are given.

\end{example}

\begin{center}
{\scriptsize {Table 1: The coset leaders and its syndromes.}}
{\small \centering}
\begin{tabular}{|c|c|c|c|}
  \hline
  0     & 0                          & $x^6$  & $( - 2 - i)x + 3i$ \\

  1     & 1                          & $x^7$  & $
3ix + (2 + i)
$ \\
  $x$   & $x$                        & $x^8$  & $
( - 1 + 2i)x + 2i
$\\
  $x^2$ & $( - 1 + 2i)x + (2 - i)$   & $x^9$  & $
2ix + (1 - 2i)
$ \\
  $x^3$ & $(2 - i)x + (1 - 2i)$      & $x^10$ & $-1$ \\
  $x^4$ & $( - 2 - i)x + ( - 1 - i)$ & $x^11$ & $x^{11}  = x^{10} x
$ \\
  $x^5$ & $( - 1 - i)x + (2 + i)$    & $x^12$ & $
x^{12}  = x^{10} x^2
$ \\
  \hline
\end{tabular}

\end{center}

\begin{center}
{\scriptsize {Table 2: The coset leaders and its syndromes.}}
{\small \centering}

\begin{tabular}{|c|c|}
  \hline
  0      & 0 \\

  1      & 1 (and its associates) \\
  $x$    & $3+i$ (and its associates) \\
   $x^2$ & $4-i$ (and its associates) \\
   $x^3$ & $2-2i$ (and its associates)   \\
  \hline
\end{tabular}

\end{center}

\end{document}